\documentclass[journal,twoside]{IEEEtran}
\usepackage{graphicx,amsmath,amssymb}
\usepackage[noadjust]{cite}
\usepackage{array}
\usepackage{cases}
\newtheorem{theorem}{Theorem}[section]
\newtheorem{proposition}{Proposition}
\newtheorem{corollary}[theorem]{Corollary}
\newtheorem{lemma}[theorem]{Lemma}

\begin{document}
\title{A Combinatorial Approach to X-Tolerant Compaction Circuits}
\author{Yuichiro Fujiwara
        and Charles J. Colbourn
\thanks{This work was supported in part by JSPS Research Fellowships for Young Scientists (YF) and by DOD grants N00014-08-1-1069 and N00014-08-1-1070 (CJC).}
\thanks{Y. Fujiwara is with the Department of Mathematical Sciences, Michigan Technological University, Houghton,
MI 49931 USA. {\tt  yfujiwar@mtu.edu}.}%
\thanks{C. J. Colbourn is with the School of Computing, Informatics, and Decision Systems Engineering,
Arizona State University, Tempe, AZ 85287-8809 USA. {\tt charles.colbourn@asu.edu}}}
\markboth{IEEE transactions on Information Theory,~Vol.~x, No.~xx,~month~year}
{Fujiwara and Colbourn: A Combinatorial Approach to X-Tolerant Compaction Circuits}


\maketitle

\begin{abstract}
Test response compaction for integrated circuits (ICs) with scan-based design-for-testability (DFT) support in the presence of unknown logic values (Xs) is investigated from a combinatorial viewpoint.
The theoretical foundations of X-codes, employed in an X-tolerant compaction technique called X-compact, are examined. 
Through the formulation of a combinatorial model of X-compact, novel design techniques are developed for X-codes to detect a specified maximum number of  errors in the presence of a specified maximum number of unknown logic values, while requiring only small fan-out.
The special class of X-codes that results leads to an avoidance problem for configurations in combinatorial designs. 
General design methods and nonconstructive existence theorems to estimate the compaction ratio of an optimal X-compactor are also derived.
\end{abstract}


\IEEEpeerreviewmaketitle

\section{Introduction}
\PARstart{T}{his} work discusses a class of codes that arise in data volume compaction of responses from integrated circuits (ICs) under scan-based test.
We first recall briefly the background of the X-tolerant compaction technique in digital circuit testing.

Digital circuit testing applies test patterns to a circuit under test and monitors the circuit's responses to the applied patterns.
A tester compares the observed response to a test pattern to the expected response and, if there is a mismatch, declares the circuit chip defective.
Usually the expected responses are obtained through fault-free simulation of the chip.

Test cost for traditional scan-based testing is dominated by test data volume and test time \cite{MCC}.
Therefore various test compression techniques have been developed to reduce test cost.
One way to achieve this is to reduce test application time and the number of test patterns by employing automatic test pattern generation (ATPG)
(see \cite{ITVLSITEST1, ITVLSITEST2, ITVLSITEST3, ITVLSITEST4} and references therein).
We are interested in the other kind of technique, using methods to hash responses while maintaining test quality.
Signature analyzers (e.g., \cite{MISR,MISR2,MISR3,OPMISR,OPMISR2}) are vulnerable to error masking caused by unknown logic values (Xs)  \cite{XTUTORIAL}. 
{\it X-compact} has been proposed in order to conduct reliable testing in the presence of Xs \cite{XCOMPACTION}.
A response compaction circuit based on X-compact is an {\it X-compactor}.
X-compactors have proved their high error detection ability in actual  systems \cite{XTEST,XTUTORIAL}.

An X-compactor can be written in matrix form as an {\it X-code} \cite{XCODES}.
Basic properties of X-codes have been studied \cite{XCODES,XCODES2}.
Graph theoretic techniques have been employed to minimize fan-out of inputs \cite{XCODESCAGE}; in general an X-compactor tolerates the presence of Xs in exchange for large fan-out.
These studies  focus on particular classes of X-codes rather than the general coding theoretic aspects.

The purpose of the present paper is to investigate theoretical foundations of X-codes and to provide general construction techniques.
In Section II we outline the combinatorial requirements for the X-compact technique and present an equivalent definition of X-codes in order to investigate X-compactors as codes and
combinatorial designs. 
Some basic properties of X-codes are also presented.
In Section III we investigate X-codes that require only small fan-out and have good error detectability and X-tolerance.
We prove the equivalence between a class of Steiner $t$-designs and particular X-codes having the maximum number of codewords and the minimum fan-out. 
This allows us to give constructions and to show existence of such X-codes.
Section IV deals with existence of X-codes in the more general situation.
Both constructive and nonconstructive theorems are provided.
Finally we conclude in Section V.

\section{Combinatorial Requirements and X-Codes}

We do not describe scan-based testing and response compaction in detail here, instead referring the reader to  \cite{XTUTORIAL,XCOMPACTION}.

Scan-based testing repeatedly applies vectors  of test inputs to the circuit, and for each test captures a vector from $\{0,1\}^n$ as the test output.  Naturally it is important that the test output be the correct one.  To determine this, the function of the circuit is simulated (in a fault-free manner) to produce a reference output.  When the test and reference outputs agree, no fault has been detected.  The first major obstacle is that fault-free simulation may be unable to determine whether a specific output is 0 or 1, and hence it is an unknown logic value X.  The second is that if each output requires a separate pin on the chip, the number of tests that can be accommodated is quite restricted.  We deal with these two problems in turn.

We define an algebraic system to describe the behavior of Xs.
The {\it X-algebra} ${\mathbb X}_2 = (\{0,1,\mbox{X}\},+,\cdot)$ over the field ${\mathbb F}_2$ is
the set $\{0,1\}$ of elements of ${\mathbb F}_2$ and a third element $\mbox{X}$, equipped with two binary operations ``$+$" (addition)  and ``$\cdot$" (multiplication) satisfying:

\begin{enumerate}
\item   $a+b$ and  $a \cdot b$ are performed in ${\mathbb F}_2$ for $a,b \in {\mathbb F}_2$;
\item $a+\mbox{X}=\mbox{X}+a=\mbox{X}$ for $a \in {\mathbb F}_2$;
\item $0\cdot \mbox{X}=\mbox{X}\cdot 0=0$ for the additive identity $0$; 
\item  $1 \cdot \mbox{X}=\mbox{X}\cdot 1 =\mbox{X}$.
\end{enumerate}

The element X is termed an {\em unknown logic value}. 

Now consider a test output ${\sf b} = (b_1,\dots,b_n) \in \{0,1\}^n$ and a reference output ${\sf c} = (c_1,\dots,c_n) \in \{0,1,\mbox{X}\}^n$.  When $c_i \in \{0,1\}$, the test and reference outputs agree on the $i$th bit when $b_i=c_i$; otherwise the $i$th bit is an {\sl error bit}. When $c_i = \mbox{X}$, whatever the value of $b_i$, no error is detected.  Thus the $i$th bit is (known to be) in error if and only if $b_i + c_i = 1$, using addition in ${\mathbb X}_2$.

Turning to the second problem, an {\sl X-compact matrix} is an $n \times m$ matrix $H$ with elements from $\{0,1\}$.  The  {\it compaction ratio} of $H$ is $n/m$.
The number of $1$s in the $i$th row is  the {\it weight}, or {\it fan-out}, of  row $i$.
Output (or response) compaction is performed by computing the vector ${\sf d} = (d_1,\dots,d_m) = {\sf b} H$  for output (arithmetic is in ${\mathbb X}_2$).  In the same way, the reference output can be compacted using the same matrix to form ${\sf r} = (r_1,\dots,r_m) = {\sf c} H$.  
As before, if $d_i \neq r_i$ and $r_i \neq \mbox{X}$ (that is, if $d_i + r_i = 1$), an error is detected.

To be of practical value, an X-compact matrix $H$ should detect the presence of error bits in ${\sf b}$ with respect to ${\sf c}$ given the compacted vectors ${\sf d}$ and ${\sf r}$ under `reasonable' restrictions on the number of errors and number of unknown logic values.  

Suppose that $b_\ell + c_\ell = 1$ (so that there is a fault to be detected).  In principle, whenever $h_{\ell j} = 1$, the fault could be observed on output $j$.  Let $L = \{ j : h_{\ell j} = 1\}$.
Suppose then that $j \in L$.  If it happens that $\sum_{i=1}^n c_i h_{ij} = \mbox{X}$, the error at position $\ell$ is {\sl masked} for output $j$ (that is, $d_j + r_j  = \mbox{X}$, and no error is observed). On the other hand, if  \[ d_j + r_j = \sum_{i=1}^n b_i h_{ij} + \sum_{i=1}^n c_i h_{ij} = \sum_{i=1}^n (b_i+c_i) h_{ij} = 0 \]
then no error is observed.  
This occurs when there are an {\sl even} number of values of $i$ for which $h_{ij} = 1$ and $b_i + c_i = 1$; because this holds when $i=\ell$ by hypothesis, the error at position $\ell$ is {\sl canceled} for output $j$ when the number of such errors is even.  When an error is masked or canceled for every output $j \in L$, it is not detected.  Otherwise, it is detected by an output that is neither masked nor canceled.

Treating X's as erasures and using traditional codes can increase the error detectability of an X-compactor \cite{ICOMPACT}.
Unfortunately, this involves postprocessing test responses and cannot be easily implemented \cite{XCOMPACTION}.
Therefore, we focus on X-compaction in which
an error is only detected by the simple comparison described here. 

There are numerous criteria in defining a ``good" X-compact matrix.  It should have a high compaction ratio and be able to detect any faulty circuit behavior anticipated in actual testing.
Power requirements, compactor delay, and wireability dictate that the weight of each row in a matrix be small to meet practical limitations on fan-in and fan-out \cite{XCODESCAGE,XTUTORIAL}.

The fundamental problem in X-tolerant response compaction is to design an X-compact matrix with  large compaction ratio that  detects  faulty circuit behavior.
To achieve this,  X-codes (which represent X-compact matrices) were introduced  \cite{XCODES}.
In this section, we discuss basic properties of X-codes.
In order to investigate  X-codes from  coding and design theoretic views, we introduce an equivalent definition.

Consider two $m$-dimensional vectors $\boldsymbol{s}_1= (s_1^{(1)},s_2^{(1)},\dots,s_m^{(1)})$ and $\boldsymbol{s}_2 = (s_{1}^{(2)},s_{2}^{(2)},\dots,s_{m}^{(2)})$,  where $s_i^{(j)} \in {\mathbb F}_2$.
The {\it addition} of 
$\boldsymbol{s}_1$ and $\boldsymbol{s}_2$
is bit-by-bit addition, denoted by
$\boldsymbol{s}_1 \oplus \boldsymbol{s}_2$; that is,
\[
\boldsymbol{s}_1 \oplus \boldsymbol{s}_2 = (s_1^{(1)}+s_{1}^{(2)},
s_2^{(1)}+s_{2}^{(2)},\dots,s_m^{(1)}+s_{m}^{(2)}).
\]
The {\it superimposed sum} of  $\boldsymbol{s}_1$
and $\boldsymbol{s}_2$,
denoted $\boldsymbol{s}_1 \vee \boldsymbol{s}_2$, is 
\[
\boldsymbol{s}_1 \vee \boldsymbol{s}_2 = (s_1^{(1)} \vee s_{1}^{(2)},
s_2^{(1)} \vee s_{2}^{(2)},\dots,s_m^{(1)} \vee s_{m}^{(2)}),
\]
where $s_i^{(j)} \vee s_k^{(l)} = 0$ if $s_i^{(j)} = s_k^{(l)} =0$, otherwise $1$.
An $m$-dimensional vector $\boldsymbol{s}_1$ {\it covers} an
$m$-dimensional vector $\boldsymbol{s}_2$ if
$\boldsymbol{s}_1 \vee \boldsymbol{s}_2 = \boldsymbol{s}_1$.

For a finite set $S=\{\boldsymbol{s}_1,\dots,\boldsymbol{s}_s\}$
of $m$-dimensional vectors, define
\[\bigoplus S = \boldsymbol{s}_1\oplus\dots\oplus\boldsymbol{s}_s
\mbox{\ and\ } 
\bigvee S = \boldsymbol{s}_1\vee\dots\vee\boldsymbol{s}_s.\]
When $S=\{\boldsymbol{s}_1\}$ is a singleton, 
 $\bigoplus S = \bigvee S= \boldsymbol{s}_1$.
For $S = \emptyset$ we define $\bigoplus S = \bigvee S= \boldsymbol{0}$,
the zero vector.

Let $d$ be a positive integer and $x$ a nonnegative integer.
An $(m,n,d,x)$ {\it X-code} ${\mathcal X}= \{\boldsymbol{s}_1,\boldsymbol{s}_2,\dots,\boldsymbol{s}_n\}$
is a set of $m$-dimensional vectors over ${\mathbb F}_2$ such that $|{\mathcal X}|=n$ and 
\[(\bigvee S_1) \vee (\bigoplus S_2) \not= \bigvee S_1.\]
for any pair of mutually disjoint  subsets $S_1$ and $S_2$ of ${\mathcal X}$ with $|S_1|=x$ and $1 \leq |S_2| \leq d$.
A vector $\boldsymbol{s}_i \in {\mathcal X}$ is a {\it codeword}.
The {\it weight} of a codeword $\boldsymbol{s}_i$ is  $|\{s_j^{(i)} \not= 0: s_j^{(i)} \in \boldsymbol{s}_i \}|$.
The ratio $n/m$ is the {\it compaction ratio} of ${\mathcal X}$.
An X-code forming an orthonormal basis of the $m$-dimensional linear space over ${\mathbb F}_2$ is {\it trivial}.

Roughly speaking, an X-code is a set of codewords such that for every positive integer $d' \leq d$ no superimposed sum of any $x$ codewords covers the vector obtained by adding up any $d'$ codewords chosen from the rest of the $n-x$ codewords.
Now we present a method of designing an X-compact matrix from an X-code.

\begin{proposition}\label{equivalence}
There exists an $(m,n,d,x)$ X-code ${\mathcal X}$ if and only if there exists an
$n \times m$ X-compact matrix $H$ which detects any combination of $d'$ faults ($1 \leq d' \leq d$) in the presence of at most $x$ unknown logic values.
\end{proposition}

\begin{proof}
First we prove necessity. Assume that ${\mathcal X}$ is an $(m,n,d,x)$ X-code.
Write
${\mathcal X}=\{\boldsymbol{s}_1,\boldsymbol{s}_2,\dots,\boldsymbol{s}_n\}$, where $\boldsymbol{s}_i = (s_1^{(i)},s_2^{(i)},\dots,s_m^{(i)})$ for $1 \leq i \leq n$.
Define an $n \times m$ matrix $H=(h_{i,j})$ as $h_{i,j} = s_j^{(i)}$.
We show that $H$ forms an X-compact matrix that detects a fault if the test output ${\sf b}$ contains $d'$ error bits, $1 \leq d' \leq d$, and up to $x$ Xs.

Let $E = \{k : b_k + c_k = 1\}$, the set of indices of error bits, have cardinality $d'$. 
Let $X = \{k : c_k = \mbox{X}\}$, the set of indices of unknown logic values, have cardinality $x$.
Now comparing $d_\ell$ and $r_\ell$, 
\begin{eqnarray}\label{condition1}
d_\ell + r_\ell &=&
\sum_{k}b_k\cdot h_{k,\ell}+\sum_{k}c_k\cdot h_{k,\ell}\nonumber\\
&=& \sum_{k\in E,X}(b_k+c_k)\cdot h_{k,\ell}\nonumber\\
&=& \sum_{k\in E}\mbox{1}\cdot h_{k,\ell}+\sum_{k\in X}\mbox{X}\cdot h_{k,\ell},
\end{eqnarray}
with operations performed in ${\mathbb X}_2$.
Because the set of rows of $H$ forms the set of codewords of ${\mathcal X}$,
no superimposed sum of $x$ rows covers the vector obtained by an addition of any $d'$ rows.
Hence there exists a column $c$ such that
\begin{equation}\label{condition2}
\sum_{k\in E}\mbox{1}\cdot h_{k,c} =
1 \mbox{ and } \sum_{k\in X}\mbox{X}\cdot h_{k,c} = 0.
\end{equation}
Then (\ref{condition1}) and (\ref{condition2}) imply $d_c+r_c = 1$, that is, $H$ detects a fault.

Because (\ref{condition2}) holds if and only if the right hand side of (\ref{condition1}) equals one for $l=c$, sufficiency is straightforward.
\end{proof}

By virtue of this equivalence,
we can employ various known results and techniques in coding theory
to design an X-compactor with good error detection ability, X-tolerance,
and a high compaction ratio.
For the case when $x=0$, an $(m,n,d,0)$ X-code forms an $n \times m$ X-compact matrix
which is a parity-check matrix of a binary linear code of length $n$
and minimum distance $d$.
In fact, since the condition that $x=0$ implies the absence of Xs,
this special case is reduced to traditional space compaction. 
Because our focus is compaction in the presence of unknown logic values,
we assume that $x \geq 1$ henceforth unless otherwise stated.
In  the absence of Xs, see
\cite{ECCCOMPUTERSYSTEMS,SPACETIMECOMPACTION}.

By definition, an $(m,n,d,x)$ X-code, $d \geq 2$, is also an $(m,n,d-1,x)$ X-code.
Also an $(m,n,d,x)$ X-code forms an $(m,n,d,x-1)$ X-code.
Moreover, an $(m,n,d,x)$ X-code is an $(m,n,d+1,x-1)$ X-code \cite{XCODES}.

It can be difficult to design an X-compactor having both
the necessary error detectability and the exact number of inputs needed.
One trivial solution is to discard codewords from a larger X-code
with sufficient error detection ability and X-tolerance.
The following is another simple way to adjust the number of inputs.

\begin{proposition}\label{prop:extension}
If an $(m,n,d,x)$ X-code and an $(m',n',d',x')$ X-code exist,
there exists an $(m+m',n+n',\min\{d,d'\},\min\{x,x'\})$ X-code.
\end{proposition}

\begin{proof}
Let ${\mathcal X} = \{\boldsymbol{s}_1,\dots,\boldsymbol{s}_n\}$
be an $(m,n,d,x)$ X-code and
${\mathcal Y} = \{\boldsymbol{t}_1,\dots,\boldsymbol{t}_{n'}\}$
an $(m',n',d',x')$ X-code.
Extend each codeword $\boldsymbol{s}_i = (s_1^{(i)},\dots,s_m^{(i)})$ of ${\mathcal X}$
by appending $m'$ $0$'s so that extended vectors have the form
$\boldsymbol{s}_i' = (s_1^{(i)},\dots,s_m^{(i)},0,\dots,0)$.
Similarly extend each codeword $\boldsymbol{t}_j = (t_1^{(j)},\dots,t_{m'}^{(j)})$ of
${\mathcal Y}$ by appending  $m$ $0$s
so that extended vectors have the form
$\boldsymbol{t}_j' = (0,\dots,0,t_1^{(j)},\dots,t_{m'}^{(j)})$.
The extended $(m+m')$-dimensional vectors form
an $(m+m',n+n',\min\{d,d'\},\min\{x,x'\})$ X-code.
\end{proof}

Proposition \ref{prop:extension} says that given an $(m,n,d,x)$ X-code,
a codeword of weight less than or equal to $x$
does not essentially contribute to the compaction ratio (see also \cite{XCODES}).
In fact, if ${\mathcal X}$ contains such a codeword
$\boldsymbol{s}_i=(s_1^{(i)},\dots,s_m^{(i)})$,  there exists
at least one coordinate $m'$ such that  $s_{m'}^{(i)}=1$ and  $s_{m'}^{(j)} = 0$ for any other codeword
$\boldsymbol{s}_j \in {\mathcal X}$.
Hence we can delete $\boldsymbol{s}_i$ and 
coordinate $m'$ from ${\mathcal X}$ while keeping $d$ and $x$.
By applying Proposition \ref{prop:extension} and combining a trivial X-code and another X-code,
we can obtain an X-code having the same number of codewords
with  compaction ratio no smaller.
For this reason, when constructing an $(m,n,d,x)$ X-code explicitly,
we assume that every codeword has weight greater than $x$.

Let $M(m,d,x)$ be the maximum number $n$ of codewords
for which there exists an $(m,n,d,x)$ X-code.
More codewords means a higher compaction ratio.
Hence an $(m,n,d,x)$ X-code satisfying $n = M(m,d,x)$ is  {\it optimal}.

Determining the exact value of $M(m,d,x)$ seems difficult except for $M(m,1,1)$.
As pointed out in \cite{XCODES},
a special case of $M(m,d,x)$ has been extensively studied
in the context of superimposed codes  \cite{SUPERIMPOSED}.
An $(1,x)$-{\it superimposed code} of size $m \times n$ is an $m \times n$
matrix $S$ with entries in ${\mathbb F}_2$ such that
no superimposed sum of any $x$ columns of $S$  covers
any other column of $S$. Superimposed codes are also called
{\it cover-free families} and {\it disjunct matrices}.

By definition, a $(1,x)$-{\it superimposed code} of size $m \times n$
is equivalent to the transpose of an X-compact matrix obtained from
an $(m,n,1,x)$ X-code. Hence known results on the maximum ratio
$n/m$ for superimposed codes immediately give information about $M(m,1,x)$.
For completeness, we list useful results on $M(m,1,x)$.

By Sperner's theorem, 

\begin{theorem} (see {\rm \cite{SPERNER,UPPERBOUNDccfSTINSON}})
For $m \geq 2$ an integer,
\[M(m,1,1) \leq {{m}\choose{\lfloor m/2 \rfloor}}.\]
\end{theorem}

Indeed by taking all the $m$-dimensional vectors of weight $\lfloor m/2 \rfloor$
as codewords, we attain the bound.
The same argument is also found in \cite{XCODES}.

The following is a simple upper bound on $M(m,1,x)$:

\begin{theorem}{\rm \cite{UPPERBOUNDccfSTINSON}}\label{m1x}
For any $x \geq 2$,  \[\log_2M(m,1,x) \leq \frac{cm\log_2x}{x^2}\] for some constant $c$.
\end{theorem}

Several different proofs of Theorem \ref{m1x} are known.
Bounds on the constant $c$ are  approximately
two in \cite{SCUPPER2}, approximately four in \cite{SCUPPER4},
and approximately eight in \cite{SCUPPER8}.

The asymptotic behavior of the maximum possible number of codewords
has been also investigated for superimposed codes. 
Define the ratio $R(x)$ as
\[R(x)=\underset{m\rightarrow\infty}{\overline{\lim}}\frac{\log_2M(m,1,x)}{m}.\]

The best lower bound $\underline{R}(x) \leq R(x)$ can be found in \cite{ASLOWER}
and the best upper bound $\overline{R}(x) \geq R(x)$ in \cite{SCUPPER2}.
The descriptive asymptotic form of the best bounds as $x \rightarrow \infty$ is
\[\underline{R}(x) \sim \frac{1}{x^2\log_2e} \mbox{\ and \ }
\overline{R}(x) \sim \frac{2\log_2x}{x^2},\]
where $e$ is Napier's constant.
For a detailed summary of the known lower and upper bounds,
see \cite{SUPERIMPOSEDBOOK}.
Constructions with many codewords
have  been  studied in  \cite{SCFU2006,SCIT2000}.
See also \cite{MACULASCDISCMATH,MACULASCDISCAPPLMATH,
DISJUNCTLOVASZLOCAL,GENERALIZEDSCTHCS}
and references therein.

\section{X-Compactors with Small Fan-Out}
In this section we consider an X-compactor having sufficient
tolerance for errors and Xs, a high compaction ratio, and small fan-out.
This section is divided into four parts.
Subsection \ref{fanouttwo} deals with background and known results
of the fan-out problem in X-compactors. Then in Subsection \ref{xcodefromSTS}
we investigate X-codes that tolerate up to two X's and have the minimum fan-out.
X-Codes with further error detection ability and X-tolerance
are investigated in Subsection \ref{hightolerance-smallfanout}.
In Subsection \ref{summary_of_Section_III} we give a brief overview of the performance of our X-codes given in this section
and compare them with other codes.

\subsection{Background: Fan-Out in X-Codes}\label{fanouttwo}

X-compact reduces the number of
bits in the compacted output while keeping error detection ability by propagating each
single bit to many signal lines. In fact, each output of the X-compactor
in \cite{XCOMPACTION} connects to about half of all inputs. However,
larger fan-in increases power requirements, area, and delay \cite{XCODESCAGE}.
When these disadvantages are concerns,
fan-out of inputs of a compactor should be small to reduce fan-in values.

In terms of X-codes, the required fan-out of input $i$ in an X-compactor
is the weight of codeword $\boldsymbol{s}_i$ of the X-code.
Hence, in order to address the fan-out problem,
it is desirable for a codeword to have small weight.
However, as mentioned in Section I, an $(m,n,d,x)$ X-code containing
a codeword with weight at most $x$ is not essential
in the sense of the compaction ratio.
Hence, throughout this section, we restrict ourselves
to $(m,n,d,x)$ X-codes in which
every codeword has weight precisely $x+1$,
namely {\it constant weight} codes.

When a compactor is required to tolerate only a single unknown logic value,
fan-out is minimized when every codeword of an X-code has constant weight two.
This extreme case was addressed in \cite{XCODESCAGE} by considering a simple graph.
We briefly restate their theorems in terms of X-codes.

A {\it graph} $G$ is a pair $(V,{\mathcal E})$ such that $V$ is a finite set
and ${\mathcal E}$ is a set of pairs of distinct elements of $V$.
An element of $V$ is called a {\it vertex}, and an element of ${\mathcal E}$
is called an {\it edge}.
The {\it girth} $g$ of $G$ is the minimal size $|C|$ of
a subset $C \subset {\mathcal E}$
such that each vertex appearing in $C$ is contained in exactly two edges.

The {\it edge-vertex incidence matrix} $H$ of a graph $G = (V,{\mathcal E})$ is a
$|{\mathcal E}|\times |V|$ binary matrix $H = (h_{i,j})$ such that
rows and columns are indexed by edges and vertices respectively
and $h_{i,j}=1$ if the $i$th edge contains the $j$th vertex, otherwise $0$.
By considering the edge-vertex incidence matrix of a graph and Proposition \ref{equivalence},
we obtain:

\begin{theorem}{\rm \cite{XCODESCAGE}}\label{CAGE1}
There exists a graph $G=(V,{\mathcal E})$ of girth $g$ if and only if
there exists a $(|V|,|{\mathcal E}|,g-2,1)$ X-code of constant weight two.
\end{theorem}

\begin{theorem}{\rm \cite{XCODESCAGE}}\label{CAGE2}
A set ${\mathcal X}$ of $m$-dimensional vectors is
an $(m,n,d-1,1)$ X-code of constant weight two if and only if
it is an $(m,n,d,0)$ X-code of weight two.
\end{theorem}

These two theorems say that in order to design an X-compactor with high
error detection ability, we only need to find a graph with large girth.
The same argument is also found in \cite{XCODES}.
For existence of such graphs and more details on X-codes of constant weight two,
see \cite{XCODESCAGE} and references therein.

\subsection{Two X's and Fan-Out Three}\label{xcodefromSTS}

Multiple X's can occur; here we present  X-codes
that are tolerant to two X's and have the maximum compaction ratio.
To  accept up to two unknown logic values, we need an X-code of constant weight three.
We employ a well-known class of combinatorial designs.

A {\it set system} is an ordered pair $(V,{\mathcal B})$ such that $V$ is
a finite set of {\it points}, and ${\mathcal B}$ is a family
of subsets ({\it blocks}) of $V$.
A {\it Steiner $t$-design} $S(t,k,v)$ is a set system $(V,{\mathcal B})$,
where $V$ is a finite set of cardinality $v$ and
${\mathcal B}$ is a family of $k$-subsets of $V$
such that each $t$-subset of $V$ is contained in exactly one block.
Parameters $v$ and $k$ are the {\it order} and {\it block size}
of a Steiner $t$-design. When $t=2$ and $k=3$,
an $S(2,3,v)$ is a {\it Steiner triple system} of order $v$, STS$(v)$. An STS$(v)$ exists
if and only if $v \equiv 1, 3$ (mod $6$) \cite{TRIPLESYSTEMS}.
A {\it triple packing} of {\it order} $v$ is a set system $(V,{\mathcal B})$
such that ${\mathcal B}$ is a family of triples of a finite set $V$ of cardinality $v$
and any pair of elements of $V$ appear in ${\mathcal B}$ at most once.
An STS$(v)$ is a triple packing
of order $v\equiv 1,3$ (mod $6$) containing the maximum number of triples.

The {\it point-block incidence matrix} of a set system $(V,{\mathcal B})$ is
the binary $|V|\times |{\mathcal B}|$ matrix $H = (h_{i,j})$ such that
rows are indexed by points, columns are indexed by blocks,
and $h_{i,j}=1$ if the $i$th point is contained in the $j$th block, otherwise $0$.
The {\it block-point incidence matrix} is its transpose.

When $d=1$, an $(m,n,1,2)$ X-code of constant weight three
is equivalent to a $(1,2)$-superimposed code of size $m\times n$
of constant column weight three.
It is well known that the point-block incidence matrix of an $S(t,k,v)$ forms
an $(1,\lceil k/(t-1)\rceil-1)$-superimposed code of size
$v\times {{v}\choose{t}}/{{k}\choose{t}}$.
Hence, by using  an STS$(v)$, we obtain for every $v \equiv 1, 3$ (mod $6$)
a $(v,v(v-1)/6,1,2)$ X-code. An upper bound on the number of codewords of
$(1,2)$-superimposed codes of constant weight $k$
is available: 

\begin{theorem}{\rm \cite{EFFCOVERFREEFAMILY1}}\label{coverfreer=2}
Let $n^k(m)$ denote the maximum number of columns
of a $(1,2)$-superimposed code such that and every column
is of length $m$ and has constant weight $k$. Then, 
\[n^{2t-1}(m) \leq n^{2t}(m+1) \leq \frac{{{m}\choose{t}}}{{{2t-1}\choose{t}}}\]
with equality  if and only if there exists a Steiner $t$-design $S(t,2t-1,m)$.
\end{theorem}

The following is an immediate consequence:

\begin{theorem}\label{xcodests}
For any $(m,n,1,2)$ X-code of constant weight three,
$n \leq \frac{m(m-1)}{6}$
with equality  if and only if there exists an STS$(m)$.
\end{theorem}

Hence for  $d=1$, $x=2$, and fan-out three,
an X-code from any STS$(v)$ has the maximum compaction ratio $(v-1)/6$.

One may ask for larger error detectability of an $(m,n,1,2)$ X-code
when one (or zero) unknown logic value is assumed.
An $(m,n,d,x)$ X-code is also an $(m,n,d+1,x-1)$ X-code, and hence
any $(m,n,1,2)$ X-code from an STS$(m)$ is also an $(m,n,2,1)$ X-code.
However, a careful choice of Steiner triple systems gives higher error detectability
while maintaining the compaction ratio.

A {\it configuration} ${\mathcal C}$ in a triple packing, $(V,{\mathcal B})$, 
is a subset ${\mathcal C} \subseteq {\mathcal B}$.
The set of points appearing in at least one block of a configuration ${\mathcal C}$
is denoted by $V({\mathcal C})$.
Two configurations ${\mathcal C}$ and ${\mathcal C}'$ are
{\it isomorphic}, denoted ${\mathcal C} \cong {\mathcal C}'$,
if there exists a bijection $\phi : V({\mathcal C}) \rightarrow V({\mathcal C}')$
such that for each block $B \in {\mathcal C}$,
the image $\phi(B)$ is a block in ${\mathcal C}'$.
When $|{\mathcal C}|=i$,
a configuration ${\mathcal C}$ is  an {\it $i$-configuration}.
A configuration ${\mathcal C}$ is  {\it even}
if for every point $a$ appearing in ${\mathcal C}$
the number $|\{B: a \in B \in {\mathcal C}\}|$
of blocks containing $a$ is even.
Because every block in a triple packing has three points, no $i$-configuration for $i$ odd  is even.  

A triple packing is  {\it $r$-even-free} if
for every integer $i$ satisfying $1\leq i \leq r$
it contains no even $i$-configurations.
By definition every $r$-even-free triple packing, $r\geq 2$, is also $(r-1)$-even-free.
For an even integer $r$, an $r$-even-free triple packing is also $(r+1)$-even-free.
Every triple packing is trivially $3$-even-free.
For $v >3$ an STS$(v)$ may or may not be $4$-even-free.
Up to isomorphism, the only even $4$-configuration is the {\it Pasch} configuration.
It can be written on six points and four blocks:
$\{\{a, b, c\}, \{a, d, e\}, \{f, b, d\}, \{f, c, e\}\}$.
For the list of all the small configurations in a triple packing and more complete treatments, we refer the reader to \cite{TRIPLESYSTEMS} and \cite{stopsts}.
Because a $4$-even-free STS is $5$-even-free, an STS is $5$-even-free if
and only if it contains no Pasch configuration.

\begin{lemma}\label{5-even-free}
If there exists a $5$-even-free STS$(v)$, there exists
a $(v,v(v-1)/6,3,1)$ X-code of constant weight three. The
code is  a $(v,v(v-1)/6,5,0)$ X-code of constant weight
three.
\end{lemma}

\begin{proof}
Let $(V,{\mathcal B})$ be a $5$-even-free STS$(v)$.
For every $B_i \in {\cal B}$ define a $v$-dimensional vector
$\boldsymbol{s}_i$ such that each coordinate $s_j^{(i)} \in \boldsymbol{s}_i$
is indexed by a distinct point $j \in V$ and $s_j^{(i)}=1$
if $j \in B_i$, otherwise $0$. Then we obtain a $(v,v(v-1)/6,1,2)$ X-code
${\mathcal S} = \{\boldsymbol{s}_i: B_i \in {\cal B}\}$ of constant weight three.
We prove that ${\mathcal S}$ is a $(v,v(v-1)/6,3,1)$ X-code
that is also a $(v,v(v-1)/6,5,0)$ X-code.
By definition, for $1 \leq i \leq 5$ no $i$-configuration
${\mathcal C} \subseteq {\mathcal B}$ is even. Hence
\[\bigoplus \{\boldsymbol{s}_i: B_i \in {\mathcal C}\} \not= \boldsymbol{0}.\]
This implies that ${\mathcal S}$ is a $(v,v(v-1)/6,5,0)$ X-code.
On the other hand, since no pair of points appears twice,
for any mutually distinct blocks $B_i$, $B_j$, $B_k \in {\mathcal B}$,
\[\boldsymbol{s}_i \not= \boldsymbol{s}_j \mbox{ and }\boldsymbol{s}_i
\vee (\boldsymbol{s}_j \oplus \boldsymbol{s}_k) \not=\boldsymbol{s}_i.\]
It remains to show that no codeword in ${\mathcal S}$
covers addition of three others.
Suppose to the contrary that there exist four distinct codewords
$\boldsymbol{s}_i$, $\boldsymbol{s}_j$, $\boldsymbol{s}_k$,
and $\boldsymbol{s}_l$ such that
\[\boldsymbol{s}_i \vee (\boldsymbol{s}_j \oplus \boldsymbol{s}_k
\oplus \boldsymbol{s}_l) =  \boldsymbol{s}_i.\]
Because no pair of points appears twice and every block has exactly three points,
the only possible case is that the $4$-configuration
$\{B_i, B_j, B_k, B_l\}$ forms a Pasch configuration, and hence
it is even, a contradiction. 
\end{proof}

Steiner triple systems avoiding Pasch configurations have been long studied  as {\it anti-Pasch} STSs \cite{TRIPLESYSTEMS}.

\begin{theorem}{\rm \cite{RESOLUTIONANTIPASCH}}\label{anti-Pasch}
There exists a $5$-even-free STS$(v)$ if and only if $v \equiv 1, 3$ (mod $6$) and $v \not\in \{ 7,13\}$.
\end{theorem}

By combining Theorem \ref{anti-Pasch} and Lemma \ref{5-even-free}, we obtain:

\begin{theorem}\label{anti-Paschxcodes}
For every $v \equiv 1, 3$ (mod $6$) and $v \not\in \{ 7,13\}$,
there exists a $(v,v(v-1)/6,1,2)$ X-code of constant weight three
that is a $(v,v(v-1)/6,3,1)$ X-code and a $(v,v(v-1)/6,5,0)$ X-code.
\end{theorem}

An X-compactor designed from these can detect any odd number of errors unless there is  an unknown logic value.
One may want to take advantage of the high compaction ratio of the optimal $(m,n,1,2)$ X-codes arising from $4$-even-free STSs
when there is only a small possibility that more than two Xs occur or multiple errors happen with multiple Xs.
Our X-codes from $4$-even-free STSs also have high performance in such situations:

\begin{theorem}\label{one_error_three_xs}
The probability that a $(v,v(v-1)/6,1,2)$ X-code from a $4$-even-free STS$(v)$ fails to detect a single error when there are exactly three Xs is
$\frac{162(v-3)^2}{(v+2)(v+3)(v-4)(v^2-v-18)}$.
\end{theorem}

\begin{proof}
Because there is only one error, an X-code fails to detect this error when all three points in the block
that corresponds to the error are contained in at least one block corresponding to an X.
The number of occurrences of each $4$-configuration in an STS$(v)$ is determined by $v$ and the number of Pasch configurations (see \cite{TRIPLESYSTEMS}, for example).
A simple calculation proves the assertion.
\end{proof}

\begin{theorem}\label{two_errors_two_xs}
The probability that a $(v,v(v-1)/6,1,2)$ X-code from a $4$-even-free STS$(v)$ fails to detect errors when there are exactly two Xs and exactly two errors is
$\frac{1296}{(v+2)(v+3)(v-4)(v^2-v-18)}$.
\end{theorem}

\begin{proof}
A $(v,v(v-1)/6,1,2)$ X-code from a $4$-even-free STS$(v)$ fails to detect errors when
there are exactly two Xs and exactly two errors only when corresponding four blocks form a $4$-configuration isomorphic to $\{\{a,b,c\},\{d,e,f\},\{a,e,g\},\{c,f,g\}\}$
where the first two blocks represent Xs and the other two blocks correspond to errors. The number of occurrences of the $4$-configuration
in a $4$-even-free STS$(v)$ is $\frac{v(v-1)(v-3)}{4}$, and the total number of occurrences of all $4$-configurations is ${\frac{v(v-1)}{6}}\choose{4}$  \cite{TRIPLESYSTEMS}.
Divide $\frac{v(v-1)(v-3)}{4}$ by ${{4}\choose{2}}{{\frac{v(v-1)}{6}}\choose{4}}$ to obtain the probability that the X-code fails to detect the two errors.
\end{proof}

Hence when a $4$-even-free STS of sufficiently large order is used, the probability that the corresponding X-code fails to detect errors
when the sum of the numbers of errors and Xs is at most four is close to zero.
A more complicated counting argument is necessary to calculate the performance of X-codes from STSs 
when the sum of the numbers of errors and Xs is greater than four.
For more complete treatments and current research results on counting configurations in Steiner triple systems,
we refer the reader to \cite{stopsts} and references therein.

Useful explicit constructions for $5$-even-free STS$(v)$ can be found in 
\cite{ANTIPASCHJLMS,ANTIPASCHSTINSONDM,ANTIPASCHGRANNELLARS,
ANTIPASCHBROUWER,RESOLUTIONANTIPASCH,TRIPLESYSTEMS}.
The  cyclic $5$-sparse Steiner triple systems in \cite{MITRECOLBOURN}
provide examples of $5$-even-free STS$(v)$ for $v \leq 97$,
because cyclic $5$-sparse systems are all anti-Pasch.
Further $r$-even-freeness improves the error detectability of the resulting X-code:

\begin{theorem}\label{higherevenfreeforxcodes}
For $r\geq 4$, if there exists an $r$-even-free triple packing $(V,{\mathcal B})$,
there exists a $(|V|,|{\mathcal B}|,1,2)$ X-code of constant weight three
that is also a
$(|V|,|{\mathcal B}|,3,1)$ X-code and a $(|V|,|{\mathcal B}|,r,0)$ X-code.
\end{theorem}

\begin{proof}
Let $(V,{\mathcal B})$ be an $r$-even-free triple packing of order $v$.
For every $B_i \in {\cal B}$ define a $v$-dimensional vector
$\boldsymbol{s}_i$ such that each coordinate $s_j^{(i)} \in \boldsymbol{s}_i$
is indexed by a distinct point $j \in V$ and $s_j^{(i)}=1$
if $j \in B_i$, otherwise $0$. Then we obtain a $(|V|,|{\mathcal B}|,1,2)$ X-code
${\mathcal S} = \{\boldsymbol{s}_i: B_i \in {\cal B}\}$ of constant weight three.
It suffices to prove that ${\mathcal S}$ forms a $(|V|,|{\mathcal B}|,r,0)$ X-code.
Suppose to the contrary that ${\mathcal S}$ is not
a $(|V|,|{\mathcal B}|,r,0)$ X-code.
Then for some $r'\leq r$ there exists a set of $r'$ codewords
$\boldsymbol{s}_i$, $\boldsymbol{s}_j$\dots, $\boldsymbol{s}_{k}$ such that
\[\boldsymbol{s}_i \oplus \boldsymbol{s}_j \dots\oplus \boldsymbol{s}_k =
\boldsymbol {0}.\]
However, the set of the corresponding blocks
$B_i$, $B_j$,\dots,$B_k$ forms an even $r'$-configuration, a contradiction.
\end{proof}

One may want an $r$-even-free STS with large $r$ to obtain higher error detection ability while keeping
the maximum compaction ratio. Although it is known that every Steiner triple system has a configuration with seven or fewer blocks so that every element of the configuration belongs to at least two \cite{stopsts}, it may happen that none of these are even.  Nevertheless, the following 
gives an upper bound of even-freeness of Steiner triple systems.

\begin{theorem}\label{no8evenfreeSTS}
For $v>3$ there exists no $8$-even-free STS$(v)$.
\end{theorem}

\begin{proof}
Suppose to the contrary that there exists an STS$(v)$,
${\cal S}$, that is $8$-even free.
Consider a $4$-configuration ${\mathcal C}$ isomorphic to
$\{\{a,b,e\},\{c,d,e\},\{a,c,f\},\{b,d,g\}\}$;
the points $f$ and $g$ are each contained in exactly one block.
For any anti-Pasch STS$(v)$
the number of occurrences of configurations isomorphic to
${\mathcal C}$ is $v(v-1)(v-3)/4$ \cite{FOURLINESTS} (see also \cite{polytope}).
Because $v \geq 7$, we have $v(v-1)(v-3)/4 > {{v}\choose{2}}$. Hence
there is a pair of configurations ${\mathcal A}$ and ${\mathcal B}$ such that
${\mathcal A} \cong {\mathcal B} \cong {\mathcal C}$ 
and they share the two points  contained in exactly one block.
In other words,
there exists a pair ${\mathcal A}$ and ${\mathcal B}$ having the form
$\{\{a,b,e\},\{c,d,e\},\{a,c,f\},\{b,d,g\}\}$
and $\{\{a',b',e'\},\{c',d',e'\},\{a',c',f\},\{b',d',g\}\}$ respectively.
If there is no common block between $A$ and $B$,
then the merged configuration ${\mathcal A} \cup {\mathcal B}$ forms
an even configuration consisting of eight blocks, a contradiction.
Otherwise, there is at least one block contained in both ${\mathcal A}$ and ${\mathcal B}$.
Removing blocks shared between ${\mathcal A}$ and ${\mathcal B}$ from their union, we obtain
an even configuration on four or six blocks, a contradiction.
\end{proof}

By combining Theorems  \ref{xcodests}, \ref{higherevenfreeforxcodes}, and \ref{no8evenfreeSTS}, we have:
\begin{theorem}\label{no8evenfreeXcode}
There exists no $(m,n,1,2)$ X-code that achieves the maximum compaction ratio $(m-1)/6$ and is also an $(m,n,3,1)$ X-code and an $(m,n,8,0)$ X-code.
\end{theorem}

An STS is $7$-even-free if and only if it is $6$-even-free.
Up to isomorphism, there are two kinds of even $6$-configurations
which may appear in an STS. One is called the {\it grid} and the other is
the {\it double triangle}. Both $6$-configurations are
described by nine points and six blocks:
$\{\{a,b,c\},\{d,e,f\},\{g,h,i\},\{a,d,g\},\{b,e,h\},\{c,f,i\}\}$ and
$\{\{a,b,c\},\{a,d,e\},\{c,f,e\},\{b,g,h\},\{d,h,i\},\{f,g,i\}\}$ respectively.
By definition, an STS is $6$-even-free if it simultaneously avoids Pasches, grids, and double triangles.
We do not know whether there exists a $6$-even-free STS$(v)$ for any $v>3$.
However,  a moderately large number of
triples can be included while keeping $6$-even-freeness:

\begin{theorem}\label{6evenfree}
There exists a constant $c>0$ such that for sufficiently large $v$
there exists a $6$-even-free triple packing of order $v$ with $cv^{1.8}$
triples.
\end{theorem}

\begin{proof}
Let ${\cal C'}$ be a set of representatives of all of the nonisomorphic even configurations on
six or fewer  triples and let ${\cal C''}$ be a configuration
consisting of pair of distinct triples sharing a pair of elements.
Let ${\cal C}= {\mathcal C}' \cup {\mathcal C}''$.
Pick uniformly at random triples from $V$ with probability $p=c'v^{-\frac{6}{5}}$
independently, where $c'$ satisfies $0<c'<(\frac{10}{41\cdot79\cdot83})^{\frac{1}{5}}$.
Let $b_{\cal C}$ be a random variable counting the  configurations
isomorphic to a member of ${\cal C}$ in the resulting set of triples.
Define $E(b_{\cal C})$ as its expected value. Then

\begin{eqnarray*}
E(b_{\cal C}) &\leq&
{{v}\choose{4}}{{{4}\choose{3}}\choose{2}} p^2+
{{v}\choose{6}}{{{6}\choose{3}}\choose{4}} p^4+
{{v}\choose{9}}{{{9}\choose{3}}\choose{6}}p^6\\
&=&{{{9}\choose{3}}\choose{6}}\frac{c'^6 v^{1.8}}{9!}+f(v),
\end{eqnarray*}
where $f(v)=O(v^{1.6})$. By Markov's Inequality,
\[P\left(b_{\cal C} \geq 2E(b_{\cal C})\right)
\leq \frac{1}{2}.\]
Hence,
\[P\left(b_{\cal C}\leq 2{{{9}\choose{3}}\choose{6}}
\frac{c'^6v^{1.8}}{9!}+2f(v)\right)\geq\frac{1}{2}.\]

Let $t$ be a random variable counting the  triples
and $E(t)$ its expected value.
Then
\[E(t)=p{{v}\choose{3}}=\frac{c'}{6}v^{1.8} - g(v),\]
where $g(v)=O(v^{0.8})$. Because $t$ is a binomial random variable,
by Chernoff's inequality, for sufficiently large $v$
\begin{equation*}
P\left( t<\frac{E(t)}{2}\right) < e^{-\frac{E(t)}{8}}<\frac{1}{2}.
\end{equation*}

Hence, if $v$ is sufficiently large, then with positive probability we have
a set ${\cal B}$ of triples with the property
that $|{\cal B}| > \frac{E(t)}{2}$ and
the number of configurations in ${\cal B}$ isomorphic
to a member of ${\cal {\cal C}}$ is at most
\[2{{{9}\choose{3}}\choose{6}}\frac{c'^6v^{1.8}}{9!}+2f(v).\]
Let $ex(v,r)$ be the maximum cardinality $|{\cal B}|$ such that
there exists an $r$-even-free triple packing.
By deleting a triple from each configuration
isomorphic to a member of ${\cal C}$, we obtain
\[ex(v,6) \geq \frac{E(t)}{2}
-2{{{9}\choose{3}}\choose{6}}\frac{c'^6v^{1.8}}{9!}+h(v),\]
where $h(v)=O(v^{1.6})$.
Then
for some positive constant $c$ and sufficiently large $v$,
it holds that $ex(v,6) \geq cv^{1.8}$. 
\end{proof}

Hence we have:
\begin{theorem}\label{6evenfreexcode}
There exists a constant $c>0$ such that for sufficiently large $v$
there exists a $(v,cv^{1.8},1,2)$ X-code that is also a $(v,cv^{1.8},3,1)$ X-code and a $(v,cv^{1.8},6,0)$ X-code.
\end{theorem}

An STS$(v)$ has approximately $v^2/6$ triples.
The same technique can be used to obtain a lower bound on $ex(v,r)$ for $r\geq 8$.
In fact,  $ex(v,8)$ is at least $O(v^{\frac{12}{7}})$, and hence for sufficiently large $v$ there exists a constant $c>0$ such that
there exists a $(v,cv^{\frac{12}{7}},1,2)$ X-code that is also a $(v,cv^{\frac{12}{7}},3,1)$ X-code and a $(v,cv^{\frac{12}{7}},8,0)$ X-code.

\subsection{Higher X-Tolerance with the Minimum Fan-Out}\label{hightolerance-smallfanout}

In general, the probability that a defective digital circuit produces an error at a specific
signal output line is quite small. In fact, several errors are unlikely happen simultaneously \cite{XTUTORIAL,BISTWOHL}.
Also, multiple Xs  with errors are rare \cite{XCOMPACTION}.
Therefore, X-codes given in Theorems \ref{anti-Paschxcodes} and \ref{6evenfree} are particularly useful for relatively simple scan-based testing
such as built-in self-test (BIST) where the tester is only required to detect defective chips.
Nonetheless, more sophisticated X-codes are also useful to improve test quality and/or
to identify or narrow down the error sources by taking advantage of more detailed information
about when incorrect responses are produced \cite{XCOMPACTION}.
Hence, for use in higher quality testing and error diagnosis support,
it is of theoretical and practical interest to consider
$(m,n,d,x)$ X-codes of constant weight $x+1$, where $x\geq 3$ or $d\geq 6$.

For $(m,n,1,2)$ X-codes of constant weight three,
we employed Theorem \ref{coverfreer=2} to obtain an upper bound
on the number of codewords. 
The following theorem gives a generalized upper bound:

\begin{theorem}{\rm \cite{EFFCOVERFREEFAMILY2}}\label{coverfree=r}
Let $n(x, m, k)$ denote the maximum number of columns
of a $(1,x)$-superimposed code
such that every column is of length $m$ and has constant weight $k$. Then,
for every $x$, $t$ and $i = 0, 1$ or $i \leq x/2t^2$,
\[n(x, m, x(t-1)+1+i) \leq {{m-i}\choose{t}}/{{k-i}\choose{t}}\]
for all sufficiently large $m$,
with equality if and only if there exists a Steiner $t$-design
$S(t, x(t-1)+1, m-i)$.
\end{theorem}

By putting $t=2$ and $i=0$, we obtain:

\begin{corollary}\label{upperxcodex+1}
For an $(m,n,1,x)$ X-code of constant weight $x+1$,
\[n \leq {{m}\choose{2}}/{{x+1}\choose{2}}\]
for all sufficiently large $m$, with equality if and only if there is an $S(2,x+1,m)$.
\end{corollary}

Because the set of columns of the block-point incidence matrix of any $S(2,x+1,m)$
forms an $(m,m(m-1)/x(x+1),1,x)$ X-code of constant weight $x+1$,
the existence of Steiner $2$-designs is our next interest.
For $k\in\{4,5\}$,  necessary and sufficient conditions for existence of
an $S(2,k,v)$ are known:

\begin{theorem}{\rm \cite{HANANIK4}}
There exists an $S(2,4,v)$ if and only if $v \equiv 1,4$ (mod $12$).
\end{theorem}

\begin{theorem}{\rm \cite{HANANIK5}}
There exists an $S(2,5,v)$ if and only if $v \equiv 1,5$ (mod $20$).
\end{theorem}

For $k \geq 6$, the necessary and sufficient conditions on $v$ for existence of
an $S(2,k,v)$ are not known in general; 
the existence of a Steiner $2$-design is solved
only in an asymptotic sense \cite{WILSONPBDCLOSURE3}, although for `small' values
of $k$ substantial results are known.
For a comprehensive table of known Steiner $2$-designs, see \cite{CRCDESIGN2}.

As with X-codes from Steiner triple systems,
the error detectability can be improved by considering avoidance of
even configurations.

An $S(2,k,v)$, $(V,{\mathcal B})$, is  {\it $r$-even-free} if for $1 \leq i \leq r$ it contains
no subset ${\mathcal C} \subseteq {\mathcal B}$
such that $|{\mathcal C}|=i$ and each point appearing in ${\mathcal C}$
is contained in exactly an even number of blocks in ${\mathcal C}$.
A {\it generalized Pasch} configuration in an $S(2,k,v)$, $(V,{\mathcal B})$, is
a subset ${\mathcal C} \subset {\mathcal B}$ such that $|{\mathcal C}|=k+1$
and each point appearing in ${\mathcal C}$ is contained
exactly two blocks of ${\mathcal C}$.
As with triple systems, an $S(2,k,v)$ is $(k+1)$-even-free
if and only if it contains no generalized Pasch configurations.

\begin{theorem}\label{r-evenfreeS(2,k,v)xcodes}
If an $r$-even-free $S(2,k,v)$ for $r\geq k+1$ exists,
there exists a $(v,v(v-1)/k(k-1),1,k-1)$ X-code
of constant weight $k$ that is also a $(v,v(v-1)/k(k-1),k,1)$ X-code and a $(v,v(v-1)/k(k-1),r,0)$ X-code.
\end{theorem}

\begin{proof}
Let $(V,{\mathcal B})$ be an $r$-even-free $S(2,k,v)$.
For every $B_i \in {\cal B}$ define a $v$-dimensional vector
$\boldsymbol{s}_i$ such that each coordinate $s_j^{(i)} \in \boldsymbol{s}_i$
is indexed by a distinct point $j \in V$ and $s_j^{(i)}=1$
if $j \in B_i$, otherwise $0$. Then we obtain a $(v,v(v-1)/k(k-1),1,k-1)$ X-code
${\mathcal S} = \{\boldsymbol{s}_i: B_i \in {\cal B}\}$ of constant weight $k$.
By definition of an $r$-even-free $S(2,k,v)$, it is straightforward to see that
${\mathcal S}$ is also a $(v,v(v-1)/k(k-1),r,0)$ X-code.
It suffices to prove that
${\mathcal S}$ can also be used as a $(v,v(v-1)/k(k-1),k,1)$ X-code.
Assume that this is not the case. Then,
by following the argument in the proof of Lemma \ref{5-even-free},
${\mathcal B}$ contains a generalized Pasch configuration, a contradiction.
\end{proof}

Existence of an $r$-even-free design has been investigated in the study of
erasure-resilient codes for redundant array of independent disks (RAID) \cite{RAIDERASURECOLBOURN}. In fact,
infinitely many $r$-even-free $S(2,k,v)$s can be obtained from affine spaces
over ${\mathbb F}_q$ \cite{JIMBOERASURE}.

\begin{theorem}{\rm \cite{JIMBOERASURE}}\label{affineerasure}
For any odd prime power $q$ and positive integer $n \geq 2$
the points and lines of $AG(n,q)$ form a $(2q-1)$-even-free $S(2,q,q^n)$.
\end{theorem}

By combining  Theorems  \ref{r-evenfreeS(2,k,v)xcodes} and  \ref{affineerasure},
we obtain:

\begin{theorem}\label{AGXcodes}
For any odd prime power $q$ and positive integer $n \geq 2$,
there exists a $(q^n,q^{n-1}(q^n-1)/(q-1),1,q-1)$ X-code of constant weight $q$
that is also an
$(q^n,q^{n-1}(q^n-1)/(q-1),q,1)$ X-code and a $(q^n,q^{n-1}(q^n-1)/(q-1),2q-1,0)$ X-code.
\end{theorem}

\subsection{Characteristics of X-Codes from Combinatorial Designs}\label{summary_of_Section_III}

We have given  tight upper bounds of compaction ratio for $(m,n,1,x)$ X-codes
with the minimum fan-out and presented explicit construction methods for X-codes that attain the bounds.
As far as the authors are aware, these are the first mathematical bounds and construction techniques
for this type of optimal X-code with constant weight greater than two.
Optimal X-codes given in Theorems \ref{anti-Paschxcodes} and \ref{AGXcodes} in particular have higher error detection ability
when the number of Xs is smaller than $x$. The known construction technique using hypergraphs, briefly mentioned in \cite{XCODESCAGE},
can not guarantee the same error detection ability.

To illustrate the usefulness of our X-codes, here we compare the error detection ability of an example X-code
that can be generated using Theorem \ref{anti-Paschxcodes} with characteristics of X-codes proposed in \cite{XTUTORIAL}.
The probability that the example $(50,500,1,1)$ X-code in Table 5 in \cite{XTUTORIAL} fails to detect a single error
when there are exactly two Xs is around $4.2\times10^{-6}$. The fan-out of this code is $11$.
Our X-code from Theorem \ref{anti-Paschxcodes}, which has the same compaction ratio, has parameters $(61,610,1,2)$.
The probability that this X-code fails to detect a single error in the same situation is exactly $0$. Its fan-out is 3, which is significantly smaller.
While the multiple error detection ability of the $(50,500,1,1)$ X-code is not specified in \cite{XTUTORIAL}, our code can always detect up to three errors
when there is only one X, and up to five errors when there is no X.
By Theorem \ref{two_errors_two_xs} the probability that our $(61,610,1,2)$ X-code fails to detect errors when there are exactly two Xs and two errors is
$1.5\times10^{-6}$. Therefore, our X-code is ideal when the fan-out problem is critical and/or fault-free simulation rarely produces three or more Xs in an expected response.

Very large optimal X-codes with very high error detecting ability and compaction ratio can be easily constructed by the same method.
For example, Theorem \ref{anti-Paschxcodes} and known results on anti-Pasch STSs immediately
give a $(601,60100,1,2)$ X-code with fan-out 3 and compaction ratio 100.
This code is also a $(601,60100,3,1)$ X-code and a $(601,60100,5,0)$ X-code.
Moreover, the probability that it fails to detect errors when there are exactly two Xs and two errors (or exactly three Xs and a single error) is
around $1.6\times10^{-11}$ (or $7.3\time10^{-7}$ respectively). As far as the authors know, there have been no X-codes available
that guarantee as high error detection ability and have very small fan-out.

As Theorems \ref{anti-Paschxcodes}, \ref{one_error_three_xs}, and \ref{two_errors_two_xs} indicate,
larger X-codes designed with this method have an even higher compaction ratio and better error detection rate.
Because discarding codewords does not affect error detection ability, one may use part of a large X-code
to achieve very high test quality when compaction ratio can be compromised to an extent.

\section{X-Codes of Arbitrary Weight}

The restriction to low-weight codewords  severely limits
the compaction ratio of an X-code. Hence, when fan-in and fan-out are not of concern,
it is desirable to use X-codes with arbitrary weight.
In this section we study the compaction ratio
and construction methods of such general X-codes.

For $d=x=2$, a $(\lceil\log_2{n}\rceil(\lceil\log_2{n}\rceil+1),n,2,2)$ X-code
was constructed for any integer $n\geq 2$ \cite{XCODES}.

\begin{theorem}{\rm \cite{XCODES}}\label{mn22old}
For any optimal $(m,n,2,2)$ X-code,
$m \leq \lceil \log_2{n}\rceil(\lceil \log_2{n}\rceil+1)$.
\end{theorem}

They also gave an explicit construction method
of a $(3 \lceil \log_3{n}\rceil,n,1,3)$ X-code.
In order to give a more general construction,
we employ design theoretic techniques for arrays.
Let $n \geq w \geq 2$.
A {\it perfect hash family},  PHF$(N;u,n,w)$,
is a set ${\mathcal F}$ of $N$ functions
$f : Y \rightarrow X$ where $|Y|=u$ and $|X|=n$,
such that, for any $C \subseteq Y$ with $|C|=w$, there exists at least one function
$f \in {\mathcal F}$ such that $f|_{C}$ is one-to-one.
A PHF$(N;u,n,w)$ can be described by a $u\times N$ matrix with entries from
a set of $n$ symbols such that for any $w$ rows there exists at least one column
in which each element is distinct.

\begin{theorem}\label{perfecthashconstruction}
If an $(m,n,d,x)$ X-code and a PHF$(N;u,n,\max\{d,x\}+1)$ exist,
there exists an $(mN,u,d,x)$ X-code.
\end{theorem}

\begin{proof}
Let $H$ be a $u\times N$ $n$-ary matrix representing a PHF$(N;u,n,\max\{d,x\}+1)$.
Assign each codeword of an $(m,n,d,x)$ X-code to a distinct symbol of the PHF 
and replace each entry of $H$ by the $m$-dimensional row vector representing
the assigned codeword. Then we obtain a $u \times mN$ binary matrix $H'$.
Taking each row of $H'$ as a codeword, we obtain a set ${\mathcal X}$
of $mN$-dimensional vectors.
It suffices to show that for any two arbitrary subsets $D, X \subseteq {\mathcal X}$
satisfying $|D|=d' \leq d$, $|X|=x' \leq x$, and $D \cap X = \emptyset$, it holds that
\begin{equation}\label{xcodecondition}
(\bigvee X) \vee (\bigoplus D) \not= \bigvee X.
\end{equation}
By considering a one-to-one function in the PHF, for any $\max\{d,x\}+1$ codewords
of ${\mathcal X}$ at least one set of $m$ coordinates forms $\max\{d,x\}+1$
distinct codewords of the original $(m,n,d,x)$ X-code.
Hence, for any choice of $D$ and $X$ there exists a subset $Y \subseteq X$ of cardinality
$|Y| = \max\{0,d'+x'-(\max\{d,x\}+1)\}$ such that
at least one  set of $m$ coordinates in $D \cup (X \setminus Y)$ 
forms distinct codewords of the original $(m,n,d,x)$ X-code.
Because $|Y| \leq d'-1 < |D|$, (\ref{xcodecondition}) holds for any $D$ and $X$.
Hence, the resulting set ${\mathcal X}$ forms an $(mN,u,d,x)$ X-code.
\end{proof}

Since their introduction in \cite{PHFMEHLHORN},
much progress has been made on existence and construction techniques for
perfect hash families (see
\cite{HASHRECURSIVEJCTA2006,IPP,LOCALLEMMAARRAY,WalkerC,MTphf} for recent results).
A concise list of known results on perfect hash families is available in \cite{CRCDESIGN2}.
We can use  perfect hash families from algebraic curves over finite fields:

\begin{theorem}{\rm \cite{ALGEBRAICPHF}}\label{algebraicPHF}
For positive integers $n \geq w$, there exists an explicit construction for
an infinite family of PHF$(N;u,n,w)$ such that $N$ is $O(\log{u})$.
\end{theorem}

Indeed when $n$ is fixed, a perfect hash family with $O(\log u)$ rows can be determined in polynomial time by a greedy method \cite{PHFdens}.

By combining Theorems \ref{perfecthashconstruction} and
\ref{algebraicPHF}, we can construct infinitely many $(m,n,d,x)$ X-codes
where $m$ is $O(\log{n})$.

\begin{theorem}
For any positive integer $d$ and nonnegative integer $x$, there exists an explicit
construction for an infinite family of $(m,n,d,x)$ X-codes,
where $m$ is $O(\log{n})$.
\end{theorem}

The following is a combinatorial recursion for X-codes.

\begin{theorem}
If an $(m,n,d,x)$ X-code and an $(\ell,n,\left\lfloor\frac{d}{2}\right\rfloor,x)$ X-code exist,
there exists an $(\ell+m,2n,d,x)$ X-code.
\end{theorem}

\begin{proof}
Let ${\mathcal X} = \{\boldsymbol{s}_1,\dots,\boldsymbol{s}_n\}$
be an $(m,n,d,x)$ X-code and
${\mathcal Y} = \{\boldsymbol{t}_1,\dots,\boldsymbol{t}_{n}\}$
an $(\ell,n,\left\lfloor\frac{d}{2}\right\rfloor,x)$ X-code.
Extend each codeword $\boldsymbol{s}_i = (s_1^{(i)},\dots,s_m^{(i)})$ of ${\mathcal X}$
by appending $\ell$ $0$'s so that extended vectors have the form
$\boldsymbol{s}_i' = (s_1^{(i)},\dots,s_m^{(i)},0,\dots,0)$.
Extend each codeword $\boldsymbol{t}_i = (t_1^{(i)},\dots,t_{l}^{(i)})$ of
${\mathcal Y}$ by combining $\boldsymbol{s}_i$
so that extended vectors have the form
$\boldsymbol{t}_i' = (s_1^{(i)},\dots,s_m^{(i)},t_1^{(i)},\dots,t_{l}^{(i)})$.
Define ${\mathcal A} = \{\boldsymbol{s}_1',\dots,\boldsymbol{s}_n'\}$,
${\mathcal B} = \{\boldsymbol{t}_1',\dots,\boldsymbol{t}_n'\}$, and
${\mathcal C} = {\mathcal A}\cup{\mathcal B}$.
We prove that ${\mathcal C}$ is an  $(\ell+m,2n,d,x)$ X-code.

Take two subsets $D, X \subseteq {\mathcal C}$
satisfying $|D|=d' \leq d$, $|X|=x' \leq x$, and $D \cap X = \emptyset$.
As in the proof of Theorem \ref{perfecthashconstruction}, it suffices to show that
for any choice of $D$ and $X$ the vector obtained by adding  all the codewords in $D$
is not covered by the superimposed sum of $X$, that is, (\ref{xcodecondition}) holds.
Define a surjection $f$ of ${\mathcal C}$ to ${\mathcal X}$ as
$f : (c_1^{(i)},\dots,c_{\ell+m}^{(i)}) \mapsto (c_1^{(i)},\dots,c_{m}^{(i)})$.
Mapping all codewords of ${\mathcal C}$ under $f$ generates two copies of ${\mathcal X}$;
one is from ${\mathcal A}$ and the other is from ${\mathcal B}$.
Define a surjection $g$ of ${\mathcal C}$ to ${\mathcal Y} \cup \{\boldsymbol{0}\}$ as
$g : (c_1^{(i)},\dots,c_{\ell+m}^{(i)}) \mapsto (c_m+1^{(i)},\dots,c_{\ell+m}^{(i)})$.
By definition, $\{g(\boldsymbol{c}) : \boldsymbol{c} \in B\} = {\mathcal Y}$ and for any
$\boldsymbol{c} \in {\mathcal A}$ the image $g(\boldsymbol{c})$ is an $\ell$-dimensional zero vector.
Let $a = |D \cap A|$ and $b = |D \cap B|$.
Because ${\mathcal Y}$ is an $(\ell,n,\left\lfloor\frac{d}{2}\right\rfloor,x)$ X-code,
if $b \leq \left\lfloor\frac{d}{2}\right\rfloor$, 
\begin{equation}\label{xcodeconditionB}
g(\bigvee X) \vee g(\bigoplus D) \not= g(\bigvee X).
\end{equation}
Hence, we only need to consider the case when $b > \left\lfloor\frac{d}{2}\right\rfloor$.
Suppose to the contrary that (\ref{xcodecondition}) does not hold.
Then,
\begin{equation}\label{assumption1}
f(\bigvee X) \vee f(\bigoplus D) = f(\bigvee X).
\end{equation}
Let
\[
a' = |\{\boldsymbol{c} \in X : f(\boldsymbol{c})=f(\boldsymbol{d}),
\boldsymbol{d} \in D\cap B\}|
\]
and
\[
a'' = |\{\boldsymbol{c} \in D\cap A : f(\boldsymbol{c})=f(\boldsymbol{d}),
\boldsymbol{d} \in D\cap B\}|.
\]
Because $\{f(\boldsymbol{c}) : \boldsymbol{c} \in {\mathcal A}\} =
\{f(\boldsymbol{c}) : \boldsymbol{c} \in {\mathcal B}\} = {\mathcal X}$
and (\ref{assumption1}) holds, $b=a'+a''$. As $a+b=d'$ and $b>\left\lfloor\frac{d}{2}\right\rfloor$,
\begin{eqnarray}\label{boundofb}
b &\leq& a+a'\nonumber\\
&\leq& \left\lfloor\frac{d}{2}\right\rfloor+a'.
\end{eqnarray}
On the other hand, $|X\cap{\mathcal B}| \leq x-a'$.
Because ${\mathcal Y}$ is also an $(\ell,n,\left\lfloor\frac{d}{2}\right\rfloor+a',x-a')$ X-code,
(\ref{xcodeconditionB}) holds, a contradiction.
\end{proof}

Next, we present a simple nonconstructive existence result for $(m,n,d,x)$ X-codes.

\begin{theorem}\label{probabilisticxcodes}
Let $d$, $x$ be a positive integers. For $n \geq \max\{2d, d+x\}$, if
\[m \geq 2^{x+1}(d+x)\log{n},\]
there exists an $(m,n,d,x)$ X-code.
\end{theorem}

\begin{proof}
Let
${\mathcal X=\{\boldsymbol{s}_1,\boldsymbol{s}_2,\dots,\boldsymbol{s}_n\}}$ be a set
of $n$ $m$-dimensional vectors
$\boldsymbol{s}_i=(s_1^{(i)},s_2^{(i)},\dots,s_m^{(i)},)$
in which each entry $s_j^{(i)}$ is defined to be $1$ with probability $p=1/2$.
Let $X$ be a set of $x$ vectors of ${\mathcal X}$ and
$D_i$ a set of $i$ vectors in ${\mathcal X}\setminus X$.
Define

\begin{eqnarray*}
A(D_i,X)=
\begin{cases}
0& \mbox{if\ } (\bigvee X)\vee (\bigoplus D_i) \not= \bigvee X,\\
1& \mbox{otherwise,}
\end{cases}
\end{eqnarray*}
and let $E(A(D_i,X))$ be its expected value. Then

\begin{eqnarray*}
E(A(D_i,X)) &=& \left(1- 2^{-x}\sum_{{1\leq j \leq i}\atop{j\  \mbox{\footnotesize odd}}}{{i}\choose{j}}2^{-i}\right)^m\\
&=& (1-2^{-x-1})^m.
\end{eqnarray*}

Let \[A_{\mathcal X} =\sum_{{X \subseteq {\mathcal X}}\atop{|X|=x}} \sum_{i=1}^d \sum_{{D_i}\atop{D_i \cup X = \emptyset}}  A(D_i,X)\]
and $E(A_{\mathcal X})$ its expected value.
Then

\begin{eqnarray*}
E(A_{\mathcal X}) & = & \sum_{{X \subseteq {\mathcal X}}\atop{|X|=x}} \sum_{i=1}^d \sum_{{D_i}\atop{D_i \cup X = \emptyset}} E(A(D_i,X))\\
&=& \sum_{i=1}^d{{n}\choose{x}}{{n-x}\choose{i}}(1-2^{-x-1})^m\label{inequalityE}\\
&<& n^{d+x}(1-2^{-x-1})^m.
\end{eqnarray*}

If $E(A_{\mathcal X}) < 1$, there exists an $(m,n,d,x)$ X-code.
Taking logarithms,

\[
m > \frac{-(d+x)\log{n}}{\log{(1-2^{-x-1})}}.
\]
Hence, if
\[m \geq 2^{x+1}(d+x)\log{n} > \frac{-(d+x)\log{n}}{\log{(1-2^{-x-1})}},\]
there exists an $(m,n,d,x)$ X-code. 
\end{proof}

Hence, for any optimal $(m,n,d,x)$ X-code with $n \geq \max\{2d, d+x\}$, $m$ is at most $O(\log{n})$.
For example, by putting $d=x=2$ we know that there exists an $(m,n,d,x)$ X-code if $m \geq 32\log{n}$.
This significantly improves the upper bound in Theorem \ref{mn22old} proved in \cite{XCODES}.

\section{Conclusions}

By formulating X-tolerant space compaction of test responses combinatorially, 
an equivalent, alternative definition of X-codes has been introduced.
This combinatorial approach gives general design methods for X-codes
and bounds on the compaction ratio.
Using this model with restricted fan-out leads to well-studied objects, the Steiner $2$-designs.
These provide constructions for X-codes having sufficient error detectability,
X-tolerance,  maximum compaction ratio, and minimum fan-out.
Constant weight X-codes with high error detectability profit from a deep connection with configurations, particularly the Pasch configuration.
The combinatorial formulation of X-tolerant compaction can also be applied
in conjunction with another compaction technique (such as time compaction). If a tester wants
an X-compactor with additional properties, the necessary structure
of the compactor may be expressed in design theoretic terms.

Our formulation can also be useful for the study of higher error detectability and error diagnosis support employing the appropriate assistance from an Automatic Test Equipment (ATE) \cite{XCOMPACTION}.
For example, the compaction technique called {\it $i$-Compact} can be understood in terms of the model in Section II \cite{ICOMPACT}.

The essential idea underlying  Theorem \ref{probabilisticxcodes}
is the stochastic coding technique for X-tolerant signature analysis \cite{XMISR}.
We used a naive value $1/2$ as the probability $p$ in the proof of Theorem \ref{probabilisticxcodes}.
To obtain a better constant coefficient, $p$ should be chosen so that it minimizes the expected value
$E(A_{\mathcal X})$, that is, it should minimize
\[\sum_{i=1}^d{{n-x}\choose{i}}\left(1-\sum_{{1\leq j \leq i}\atop{j\ \mbox{\footnotesize odd}}}
{{i}\choose{j}}p^j(1-p)^{i-j+x}\right)^m.\]
While this optimization does not affect the logarithmic order in Theorem \ref{probabilisticxcodes},
it may help a tester determine the target compaction ratio and estimate the error cancellation and masking rate of an X-tolerant Multiple Input Signature Register (X-MISR) based on stochastic coding \cite{XMISR}.

In this paper we focused on space compaction.
Nevertheless,  time compaction is of great importance as well. 
We expect the combinatorial formulation developed here to provide a useful framework for exploring time compaction as well.

\section*{Acknowledgment}
A substantial part of the research was done while the first author was visiting the Department of Computer Science and Engineering of Arizona State University.
He thanks the department for its hospitality. The authors thank an anonymous referee and the editor for helpful comments and valuable suggestions.


\begin{thebibliography}{10}
\providecommand{\url}[1]{#1}
\csname url@samestyle\endcsname
\providecommand{\newblock}{\relax}
\providecommand{\bibinfo}[2]{#2}
\providecommand{\BIBentrySTDinterwordspacing}{\spaceskip=0pt\relax}
\providecommand{\BIBentryALTinterwordstretchfactor}{4}
\providecommand{\BIBentryALTinterwordspacing}{\spaceskip=\fontdimen2\font plus
\BIBentryALTinterwordstretchfactor\fontdimen3\font minus
  \fontdimen4\font\relax}
\providecommand{\BIBforeignlanguage}[2]{{%
\expandafter\ifx\csname l@#1\endcsname\relax
\typeout{** WARNING: IEEEtran.bst: No hyphenation pattern has been}%
\typeout{** loaded for the language `#1'. Using the pattern for}%
\typeout{** the default language instead.}%
\else
\language=\csname l@#1\endcsname
\fi
#2}}
\providecommand{\BIBdecl}{\relax}
\BIBdecl

\bibitem{MCC}
E.~J. McCluskey, D.~Burek, B.~Koenemann, S.~Mitra, J.~H. Patel, J.~Rajski, and
  J.~A. Waicukauski, ``Test compression roundtable,'' \emph{{IEEE} Des. Test.
  Comput.}, vol.~20, pp. 76--87, Mar./Apr. 2003.

\bibitem{ITVLSITEST1}
A.~Lempel and M.~Cohn, ``Design of universal test sequences for {VLSI},''
  \emph{{IEEE} Trans. Inf. Theory}, vol.~31, pp. 10--17, Jan. 1985.

\bibitem{ITVLSITEST2}
G.~Seroussi and N.~H. Bshouty, ``Vector sets for exhaustive testing of logic
  circuits,'' \emph{{IEEE} Trans. Inf. Theory}, vol.~34, pp. 513--522, May
  1988.

\bibitem{ITVLSITEST3}
H.~Hollmann, ``Design of test sequences for {VLSI} self-testing using {LFSR},''
  \emph{{IEEE} Trans. Inf. Theory}, vol.~36, pp. 386--392, Mar. 1990.

\bibitem{ITVLSITEST4}
G.~D. Cohen and G.~Zemor, ``Intersecting codes and independent families,''
  \emph{{IEEE} Trans. Inf. Theory}, vol.~40, pp. 1872--1881, Nov. 1994.

\bibitem{MISR}
N.~Benowitz, D.~F. Calhoun, G.~E. Alderson, J.~E. Bauer, and C.~T. Joeckel,
  ``An advanced fault isolation system for digital logic,'' \emph{{IEEE} Trans.
  Comput.}, vol. C-24, pp. 489--497, May 1975.

\bibitem{MISR2}
E.~J. McCluskey, \emph{Logic Design Principles with Emphasis on Testable
  Semi-Custom Circuits}.\hskip 1em plus 0.5em minus 0.4em\relax Englewood
  Cliffs, NJ: Prentice-Hall, 1986.

\bibitem{MISR3}
N.~R. Saxena and E.~J. McCluskey, ``Parallel signature analysis design with
  bounds on aliasing,'' \emph{{IEEE} Trans. Comput.}, vol.~46, pp. 425--438,
  Apr. 1997.

\bibitem{OPMISR}
C.~Barnhart, V.~Brunkhorst, F.~Distler, O.~Farnsworth, B.~Keller, and
  B.~Koenemann, ``{OPMISR}: The foundation for compressed {ATPG} vectors,'' in
  \emph{Proc. Int. Test Conf.}, 2001, pp. 748--757.

\bibitem{OPMISR2}
C.~Barnhart, V.~Brunkhorst, F.~Distler, O.~Farnsworth, A.~Ferko, B.~Keller,
  D.~Scott, B.~Koenemann, and T.~Onodera, ``Extending {OPMISR} beyond 10x scan
  test efficiency,'' \emph{{IEEE} Design Test Comput.}, vol.~19, pp. 65--73,
  Sep. 2002.

\bibitem{XTUTORIAL}
S.~Mitra, S.~S. Lumetta, M.~Mitzenmacher, and N.~Patil, ``X-tolerant test
  response compaction,'' \emph{{IEEE} Des. Test. Comput.}, vol.~22, pp.
  566--574, Nov. 2005.

\bibitem{XCOMPACTION}
S.~Mitra and K.~S. Kim, ``X-compact: An efficient response compaction
  technique,'' \emph{{IEEE} Trans. Comput.-Aided Design Integr. Circuits
  Syst.}, vol.~23, pp. 421--432, Mar. 2004.

\bibitem{XTEST}
S.~Mitra, S.~Kallepalli, and K.~S. Kim, ``Analysis of {X}-compact for
  industrial designs,'' Intel Corp., 2003.

\bibitem{XCODES}
S.~S. Lumetta and S.~Mitra, ``X-codes: Theory and applications of unknowable
  inputs,'' Center for Reliable and High-Performance Computing, Univ. of
  Illinois at Urbana Champaign, Tech. Rep. CRHC-03-08 (also UILU-ENG-03-2217),
  Aug. 2003.

\bibitem{XCODES2}
------, ``X-codes: Error control with unknowable inputs,'' in \emph{Proc.
  {IEEE} Intl. Symp. Information Theory}, Yokohama, Japan, June 2003, p. 102.

\bibitem{XCODESCAGE}
P.~Wohl and L.~Huisman, ``Analysis and design of optimal combinational
  compactors,'' in \emph{Proc. 21st {IEEE} VLSI Test Symp.}, April/May 2003,
  pp. 101--106.

\bibitem{ICOMPACT}
J.~H. Patel, S.~S. Lumetta, and S.~M. Reddy, ``Application of
  {S}aluja-{K}arpovsky compactors to test responses with many unknowns,'' in
  \emph{Proc. 21st {IEEE} VLSI Test Symp.}, 2003, pp. 107--112.

\bibitem{ECCCOMPUTERSYSTEMS}
T.~R.~N. Rao and E.~Fujiwara, \emph{Error-Control Coding for Computer
  Systems}.\hskip 1em plus 0.5em minus 0.4em\relax Englewood Cliffs, NJ:
  Prentice-Hall, 1989.

\bibitem{SPACETIMECOMPACTION}
K.~K. Saluja and M.~Karpovsky, ``Testing computer hardware through data
  compression in space and time,'' in \emph{Proc. Int. Test Conf.}, 1983, pp.
  83--93.

\bibitem{SUPERIMPOSED}
W.~H. Kautz and R.~R. Singleton, ``Nonrandom binary superimposed codes,''
  \emph{{IEEE} Trans. Inf. Theory}, vol.~10, pp. 363--377, Jul. 1964.

\bibitem{SPERNER}
E.~Sperner, ``Ein satz \"{u}ber {U}ntermengen einer endlichen {M}enge,''
  \emph{Math. Z.}, vol.~27, pp. 544--548, 1928.

\bibitem{UPPERBOUNDccfSTINSON}
D.~R. Stinson and R.~Wei, ``Some new upper bounds for cover-free families,''
  \emph{J. Combin. Theory Ser. A}, vol.~90, pp. 224--234, 2000.

\bibitem{SCUPPER2}
A.~G. D'yachkov and V.~V. Rykov, ``Bounds on the length of disjunctive codes,''
  \emph{Probl. Contr. Inform. Theory}, vol.~11, pp. 7--33, 1982, in Russian.

\bibitem{SCUPPER4}
Z.~F\"{u}redi, ``On $r$-cover-free families,'' \emph{J. Combin. Theory, Ser.
  A}, vol.~73, pp. 172--173, 1996.

\bibitem{SCUPPER8}
M.~Ruszink\'{o}, ``On the upper bound of the size of the $r$-cover-free
  families,'' \emph{J. Combin. Theory, Ser. A}, vol.~66, pp. 302--310, 1994.

\bibitem{ASLOWER}
A.~G. D'yachkov, V.~V. Rykov, and A.~M. Rashad, ``Superimposed distance
  codes,'' \emph{Probl. Contr. Inform. Theory}, vol.~18, pp. 237--250, 1989.

\bibitem{SUPERIMPOSEDBOOK}
D.~Z. Du and F.~K. Hwang, \emph{Combinatorial Group Testing and Its
  Applications}, 2nd~ed.\hskip 1em plus 0.5em minus 0.4em\relax Singapore:
  World Scientific, 2000.

\bibitem{SCFU2006}
H.~L. Fu and F.~K. Hwang, ``A novel use of t-packings to construct $d$-disjunct
  matrices,'' \emph{Discrete Appl. Math.}, vol. 154, pp. 1759--1762, 2006.

\bibitem{SCIT2000}
A.~G. D'yachkov, A.~J. Macula, and V.~V. Rykov, ``New constructions of
  superimposed codes,'' \emph{{IEEE} Trans. Inf. Theory}, vol.~46, pp.
  284--290, Jan. 2000.

\bibitem{MACULASCDISCMATH}
A.~J. Macula, ``A simple construction of $d$-disjunct matrices with certain
  constant weights,'' \emph{Discrete Math.}, vol. 162, pp. 311--312, 1996.

\bibitem{MACULASCDISCAPPLMATH}
------, ``Error-correcting nonadaptive group testing with $d^e$-disjunct
  matrices,'' \emph{Discrete Appl. Math.}, vol.~80, pp. 217--222, 1997.

\bibitem{DISJUNCTLOVASZLOCAL}
H.~G. Yeh, ``$d$-{D}isjunct matrices: bounds and {L}ov\'{a}sz {L}ocal
  {L}emma,'' \emph{Discrete Math.}, vol. 253, pp. 97--107, 2002.

\bibitem{GENERALIZEDSCTHCS}
A.~{De Bonis} and U.~Vaccaro, ``Constructions of generalized superimposed codes
  with applications to group testing and conflict resolution in multiple access
  channels,'' \emph{Theor. Comput. Sci.}, vol. 306, pp. 223--243, 2003.

\bibitem{TRIPLESYSTEMS}
C.~J. Colbourn and A.~Rosa, \emph{Triple Systems}.\hskip 1em plus 0.5em minus
  0.4em\relax Oxford: Oxford Univ. Press, 1999.

\bibitem{EFFCOVERFREEFAMILY1}
P.~Erd\H{o}s, P.~Frankl, and Z.~F\"{u}redi, ``Families of finite sets in which
  no set is covered by the union of two others,'' \emph{J. Combin. Theory, Ser.
  A}, vol.~33, pp. 158--166, 1982.

\bibitem{stopsts}
C.~J. Colbourn and Y.~Fujiwara, ``Small stopping sets in {S}teiner triple
  systems,'' \emph{Cryptography and Communications}, vol.~1, no.~1, pp. 31--46,
  2009.

\bibitem{RESOLUTIONANTIPASCH}
M.~J. Grannell, T.~S. Griggs, and C.~A. Whitehead, ``The resolution of the
  anti-{P}asch conjecture,'' \emph{J. Combin. Des.}, vol.~8, pp. 300--309,
  2000.

\bibitem{ANTIPASCHJLMS}
A.~C.~H. Ling, C.~J. Colbourn, M.~J. Grannell, and T.~S. Griggs, ``Construction
  techniques for anti-{P}asch {S}teiner triple systems,'' \emph{J. Lond. Math.
  Soc. (2)}, vol.~61, pp. 641--657, 2000.

\bibitem{ANTIPASCHSTINSONDM}
D.~R. Stinson and Y.~J. Wei, ``Some results on quadrilaterals in {S}teiner
  triple systems,'' \emph{Discrete Math.}, vol. 105, pp. 207--219, 1992.

\bibitem{ANTIPASCHGRANNELLARS}
M.~J. Grannell, T.~S. Griggs, and J.~S. Phelan, ``A new look at an old
  construction for {S}teiner triple systems,'' \emph{Ars Combinat.}, vol. 25A,
  pp. 55--60, 1988.

\bibitem{ANTIPASCHBROUWER}
A.~E. Brouwer, ``Steiner triple systems without forbidden subconfigurations,''
  \emph{Mathematisch Centrum Amsterdam}, ZW 104/77, 1977.

\bibitem{MITRECOLBOURN}
C.~J. Colbourn, E.~Mendelsohn, A.~Rosa, and J.~\v{S}ir\'{a}\v{n},
  ``Anti-{M}itre {S}teiner triple systems,'' \emph{Graphs Combin.}, vol.~10,
  pp. 215--224, 1994.

\bibitem{FOURLINESTS}
M.~J. Grannell, T.~S. Griggs, and E.~Mendelsohn, ``A small basis for four-line
  configurations in {S}teiner triple systems,'' \emph{J. Combin. Des.}, vol.~3,
  pp. 51--59, 1995.

\bibitem{polytope}
C.~J. Colbourn, ``The configuration polytope of $\ell$-line configurations in
  {S}teiner triple systems,'' \emph{Mathematica Slovaca}, vol.~59, no.~1, pp.
  77--108, 2009.

\bibitem{BISTWOHL}
P.~Wohl, J.~A. Waicukauski, and T.~W. Williams, ``Design of compactors for
  signature-analyzers in built-in-self-test,'' in \emph{Proc. Int. Test Conf.},
  2001, pp. 54--63.

\bibitem{EFFCOVERFREEFAMILY2}
P.~Erd\H{o}s, P.~Frankl, and Z.~F\"{u}redi, ``Families of finite sets in which
  no set is covered by the union of $r$ others,'' \emph{Israel J. Math.},
  vol.~51, pp. 75--89, 1985.

\bibitem{HANANIK4}
H.~Hanani, ``The existence and construction of balanced imcomplete block
  designs,'' \emph{Ann. Math. Statist.}, vol.~32, pp. 361--386, 1961.

\bibitem{HANANIK5}
------, ``On balanced incomplete block designs with blocks having five
  elements,'' \emph{J. Combin. Theory Ser. A}, vol.~12, pp. 184--201, 1972.

\bibitem{WILSONPBDCLOSURE3}
R.~M. Wilson, ``An existence theory for pairwise balanced designs. {III}.
  {P}roof of the existence conjectures,'' \emph{J. Combin. Theory Ser. A},
  vol.~18, pp. 71--79, 1975.

\bibitem{CRCDESIGN2}
C.~J. Colbourn and J.~H. Dinitz, Eds., \emph{Handbook of Combinatorial
  Designs}.\hskip 1em plus 0.5em minus 0.4em\relax Boca Raton, FL: Chapman \&
  Hall/CRC, 2007.

\bibitem{RAIDERASURECOLBOURN}
Y.~M. Chee, C.~J. Colbourn, and A.~C.~H. Ling, ``Asymptotically optimal
  erasure-resilient codes for large disk arrays,'' \emph{Discrete Appl. Math.},
  vol. 102, pp. 3--36, 2000.

\bibitem{JIMBOERASURE}
M.~M\"{u}ller and M.~Jimbo, ``Erasure-resilient codes from affine spaces,''
  \emph{Discrete Appl. Math.}, vol. 143, pp. 292--297, 2004.

\bibitem{PHFMEHLHORN}
K.~Mehlhorn, \emph{Data Structures and Algorithms 1}.\hskip 1em plus 0.5em
  minus 0.4em\relax Berlin, Germany: Springer, 1984.

\bibitem{HASHRECURSIVEJCTA2006}
D.~Tonien and R.~Safavi-Naini, ``Recursive constructions of secure codes and
  hash families using difference function families,'' \emph{J. Combin. Theory
  Ser. A}, vol. 113, pp. 664--674, 2006.

\bibitem{IPP}
{Tran van Trung} and S.~S. Martirosyan, ``New constructions for {IPP} codes,''
  \emph{Des. Codes Cryptgr.}, vol.~32, pp. 227--239, 2005.

\bibitem{LOCALLEMMAARRAY}
D.~Deng, D.~R. Stinson, and R.~Wei, ``The {L}ov\'{a}sz local lemma and its
  applications to some combinatorial arrays,'' \emph{Des. Codes Cryptgr.},
  vol.~32, pp. 121--134, 2004.

\bibitem{WalkerC}
R.~A. {Walker II} and C.~J. Colbourn, ``Perfect hash families: Construction and
  existence,'' \emph{Journal of Mathematical Cryptology}, vol.~1, pp. 125--150,
  2007.

\bibitem{MTphf}
S.~S. Martirosyan and {Tran van Trung}, ``Explicit constructions for perfect
  hash families,'' \emph{Des. Codes Cryptogr.}, vol.~46, no.~1, pp. 97--112,
  2008.

\bibitem{ALGEBRAICPHF}
H.~Wang and C.~Xing, ``Explicit constructions of perfect hash families from
  algebraic curves over finite fields,'' \emph{J. Combin. Theory Ser. A},
  vol.~93, pp. 112--124, 2001.

\bibitem{PHFdens}
C.~J. Colbourn, ``Constructing perfect hash families using a greedy
  algorithm,'' in \emph{Coding and Cryptology}, Y.~Li, S.~Zhang, S.~Ling,
  H.~Wang, C.~Xing, and H.~Niederreiter, Eds.\hskip 1em plus 0.5em minus
  0.4em\relax Singapore: World Scientific, 2008.

\bibitem{XMISR}
S.~Mitra, S.~S. Lumetta, and M.~Mitzenmacher, ``X-tolerant signature
  analysis,'' in \emph{Proc. Int. Test Conf.}, 2004, pp. 432--441.

\end{thebibliography}

\end{document}